\documentclass[aps,superscriptaddress,pra,amsmath,twocolumn,amssymb]{revtex4-1}%
\usepackage{amsmath, amssymb, amsthm, alltt, setspace,bbm} 
 
\usepackage{graphicx}
\usepackage[bookmarks=false]{hyperref}
\usepackage{stmaryrd}
\usepackage{graphicx}
\usepackage{amsmath}
\usepackage{amsfonts}
\usepackage{amssymb}%
\usepackage{mathrsfs} 
\usepackage{amsthm}
\usepackage{amscd}
\usepackage{url}
\usepackage[all,knot,poly]{xypic} 
\usepackage{endnotes}
\usepackage{epsfig} 

\usepackage{tikz}
\usetikzlibrary{decorations.markings}

\usepackage{enumerate} 
\usepackage{comment} 

\newcommand{\mcF}{\mathcal{F}}

\newcommand{\ot}{\otimes}

\newcommand{\by}{\! \times \!}
\newcommand{\ra}{\rightarrow}

\newcommand{\GL}{\operatorname{GL}}

\newcommand{\rank}{\operatorname{rank}}

\newcommand{\id}{\operatorname{id}}

\newcommand{\ceil}[1]{\ensuremath{\lceil #1 \rceil}}

\newcommand{\Pol}{\operatorname{Pol}}
\newcommand{\Inv}{\operatorname{Inv}}

\newcommand{\CNOT}{\operatorname{CNOT}}

\newcommand{\SWAP}{\operatorname{SWAP}}

\def\1#1{{\bf #1}}
\def\2#1{{\cal #1}}
\def\3#1{{\sl #1}}
\def\4#1{{\tt #1}}
\def\5#1{{\sf #1}}
\def\6#1{{\mathfrak #1}}
\def\7#1{{\mathbb #1}}

\def\AND{{\sf AND}}
\def\OR{{\sf OR}}

\def\SAT{{\sf 3-SAT}}

\newcommand{\be}{\begin{equation}}
\newcommand{\ee}{\end{equation}}

\newcommand{\x}{\mathbf{x}}
\newcommand{\y}{\mathbf{y}}

\newcommand{\bra}[1]{\mbox{$\langle #1|$}}
\newcommand{\ket}[1]{\mbox{$|#1\rangle$}}
\newcommand{\braket}[2]{\mbox{$\langle #1,#2\rangle$}}

\theoremstyle{plain}
\newtheorem{theorem}{Theorem}
\newtheorem{corollary}[theorem]{Corollary}

\newtheorem{lemma}[theorem]{Lemma}
\newtheorem{observation}[theorem]{Observation}

\newtheorem{question}{Question} 
\theoremstyle{definition}
\newtheorem{remark}[theorem]{Remark}
\theoremstyle{definition}
\newtheorem{defn}[theorem]{Definition}
\theoremstyle{definition}

\theoremstyle{definition}
\newtheorem{example}[theorem]{Example}
\theoremstyle{definition}

\graphicspath{{pics/}} 
\excludecomment{nicetohave}
\excludecomment{comment}

\definecolor{dred}{rgb}{.8,0.2,.2}
\definecolor{ddred}{rgb}{.8,0.5,.5}
\definecolor{dblue}{rgb}{.2,0.2,.8}
 
\begin{document}
\title{Undecidability in Tensor Network States}

\author{Jason Morton}\email{morton@math.psu.edu}\affiliation{Department of Mathematics, Pennsylvania State University, University Park PA 16802}
\author{Jacob Biamonte}\email{jacob.biamonte@qubit.org}\affiliation{Institute for Scientific Interchange, ISI Foundation,
Via Alassio 11/c, 10126 Torino, Italy}\affiliation{Centre for Quantum Technologies, National University of Singapore, Block S15, 3 Science Drive 2, Singapore 117543}

\begin{abstract} 
Recent work has examined how undecidable problems can arise in quantum information science.  We augment this by introducing three new undecidable problems stated in terms of tensor networks.  These relate to ideas of Penrose about the physicality of a spin-network representing a physical process, closed timelike curves, and Boolean relation theory.  Seemingly slight modifications of the constraints on the topology or the tensor families generating the networks leads to problems that transition  from decidable, to undecidable to even always satisfiable.  
\end{abstract}

\maketitle

As the limitations of computers are ultimately governed by the laws of physics, and as physical process can in turn be viewed as computations, it is becoming increasingly important to understand how to bridge computer science and physics.  In this setting, quantum mechanics plays a key role and could be thought of as a generalization of classical computation.  Most work has focused on developing quantum complexity theory and algorithms governed by quantum theory.  Building on the success of this endeavor, increasingly subtle ideas in computer science are finding their way into quantum physics.  An emerging theme in this regard is decidability in quantum information science \cite{eisert2011quantum,blondel2005decidable,wolf2011problems}.

Our starting place is the recent work \cite{eisert2011quantum, wolf2011problems}, wherein some undecidable problems in quantum information theory were discovered.  We augment this work with an emphasis on finding additional problems naturally phrased in terms of tensor network states. 
We address three questions.
\begin{question} \label{ques1}
Is every tensor network built from a library $\mcF$ of $n$ at most $d$-qubit quantum gates, chained together and allowing one postselection per gate, physical?
\end{question}
For large enough $n$ and $d$, 
this question is undecidable.  
Theorems \ref{thm:hiddengrandfather} and \ref{thm:rankbound} offer alternative formulations of this result, while Theorem \ref{thm:QMMP} directly addresses Question \ref{ques1}.  
Suppose one writes down a tensor network using some fixed library of primitive operations, each constructed from some combination of unitary gates, state preparations and postselected measurements.  It may be that the resulting network is physically impossible.  That is, although each individual operation in the network would be allowed in certain circumstances, their combination would be impossible.  

One way to prove this is by considering bell states and postselected bell costates.  We can then form postselected closed timelike curves (P-CTC) \cite{lloyd2011closed}.  Thus we can obtain a zero tensor by inadvertently creating a grandfather paradox (see the elementary example of a grandfather-type paradox in Figure \ref{fig:paradoxes}).  We show that, determining whether or not we can do so is undecidable.

\begin{theorem}\label{thm:hiddengrandfather}
Determining if we can construct a tensor network which contains a hidden-grandfather-paradox type contradiction is undecidable. 
\end{theorem}
 
The {\em rank} of a tensor is the least $r$ such that it can be expressed as a sum of $r$ rank-one (product) tensors.  Hastad showed that given a tensor $T$ described by a table of numbers, the associated tensor rank decision problem ``$\rank(T)\leq r$?''  is NP-hard \cite{hastad1990tensor}; in fact most tensor problems are NP-hard \cite{hillar2009most}.  Theorem \ref{thm:QMMP} implies that given a library $\mcF$ of tensors, the question of whether we can construct a tensor network state $\Psi$ on $n$ qubits such that $\rank(\Psi)\leq r$ is undecidable.  Hence, 
\begin{theorem}\label{thm:rankbound}
The constructibility of rank-bounded tensors is undecidable.
\end{theorem}

To examine carefully the boundary between decidability and undecidability, we consider the subclass of Boolean tensor network states.  These give insight into the general situation as the image of the support map.  We address the following two questions.

\begin{question} \label{ques2}
Given a finite library $\mcF$ of tensors, can a given Boolean tensor network state be constructed from $\mcF$?
\end{question}
\begin{question} \label{ques3}
Given a finite set $\mcF$ of tensors including swap and fanout, can a given Boolean tensor network state be constructed from $\mcF$?
\end{question}

When swap and fanout are not included (as in Question \ref{ques2}), the problem is undecidable.  In contrast to this, and to illustrate the subtlety of the issue, Corollary \ref{thm:decimp} explains how including two specific tensors in $\mcF$ means that a quadratic algorithm exists to solve Question \ref{ques3}.

This work is organized as follows.  In Section \ref{sec:TN} we will recall several important points of tensor network states, and restate a not-well-known early conjecture by Penrose which relates to the physicality of a tensor network representation of a quantum process.  This leads to a discussion of tensor networks with post-selected outputs in Section \ref{sec:PhysTN}.  An interesting feature of decidability is that many of the results hold when restricted to the class of Boolean quantum states.  In Section \ref{sec:Bool} we study this class, by first showing how many problems in this class are naturally decidable, including connecting the satisfiability problem with the physicality of a quantum state.  Seemingly slight modifications result in a host of problems related to Boolean tensor network states which are undecidable.

\section{Tensor networks and the grandfather paradox}\label{sec:TN}

\begin{figure}
\[ 
\SelectTips{cm}{10}
\begin{xy}<5mm,0mm>:
(0,0) *{\vcap-}; 
p+(0,.5) *{\bullet} **@{-}; 
p+(0,.1)  ?(0)*\dir{>},
{
p+(-2,2) **\crv{p+(-1,0)& p+(-2,1.5)};
p+(0,1)  **@{-}   ?(0)*\dir{>};
},
p+(0,1) *+{\sigma_x}*\frm{-} **@{-};
p+(0,1) *{\bullet} **@{-}; 
{
p+(-1,1) **\crv{p+(-.5,0)& p+(-1,.75)};
p+(0,.1)  **@{-}   ?(0)*\dir{>};
},
p+(0,1) **@{-}  ?(.8)*\dir{>},
p+(0,-.75)   *{\vcap},
p+(1,1); 
p+(0,-3.5) **@{-}  ?(.7)*\dir{>};
(3,0) *{\vcap-}; 
p+(3,0);
p+(0,.5) *{\bullet}  **@{-};
p+(0,.1)  ?(0)*\dir{>},
{p+(.5,1) **\crv{p+(.25,0)& p+(.5,.75)} ?(1)*\dir{>};
p+(0,1.9)  **@{-};
p+(0,.15); 
p+(0,.4)  **@{-};
},
{
p+(-1,1) **\crv{p+(-.5,0)& p+(-1,.75)};
p+(0,.1)  **@{-}   ?(0)*\dir{>};
p+(0,.5) *{\bullet} **@{-};
p+(-1,1) **\crv{p+(-.5,0)& p+(-1,.75)} ?(1)*\dir{>},
p+(1,1) **\crv{p+(.5,0)& p+(1,.75)} ?(1)*\dir{>},
},
(3,1.5) *{\vcap};
(7,3); 
p+(0,-3) **@{-}  ?(.7)*\dir{>};
\end{xy}
\]
\caption{(Closed timelike curves). Grandfather paradox and unproved theorem paradox tensor networks from \cite{lloyd2011closed}, with value $0$ (left) and $(1/\sqrt{2})(\ket{00}+\ket{11})$ (right).  The problem of determining whether a tensor network can be constructed from a set of generators such that it contains a grandfather-paradox-type contradiction (is the zero tensor) is undecidable.}\label{fig:paradoxes}
\end{figure}
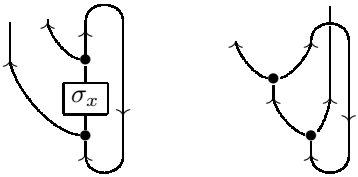

Tensor networks representing quantum states constitute a graphical representation which can be thought of as quantum circuits augmented with preparations and postselection \cite{bergholm2011categorical}. 


A grandfather-type-paradox can result from loops in a tensor network, enabling contradictions.  The tensor network formalism predicts the same result for the grandfather paradox and unproved theorem paradox tensor networks as was obtained mathematically and experimentally in \cite{lloyd2011closed}, rather than the result that would accord with Deutsch's interpretation of the  grandfather paradox network in terms of a mixed state fixed point $\rho = (\ket{0}\bra{0} + \ket{1}\bra{1})/2$ \cite{deutsch1991quantum}.  Thus \cite{lloyd2011closed} provides experimental evidence for the tensor network formalism built from unitary operators, preparations (e.g.\ into maximally entangled states), and postselection.  In particular, that a tensor network is a zero tensor iff it is forbidden by physical law.   
The following was argued by Penrose in 1967.

\begin{theorem}[Penrose, 1967 \cite{penrose1967theory}]
 The norm of a spin network vanishes iff the physical situation it represents is forbidden by the rules of quantum mechanics. 
\end{theorem}

\begin{example}[Examples of Penrose's theorem]
Consider a Bell state $\Phi^+ = \ket{00}+\ket{11}$.  The amplitude of the first party measuring $\ket{0}$ followed by the second party measuring $\ket{1}$ is zero. This vanishing tensor network contraction is given by $\braket{01}{\Phi^+}$.  A second example is found by considering the norm of a network of connected tensors thought to represent a valid state $\ket{\psi}$. If the norm found from contracting the state with a conjugated copy of itself $\braket{\psi}{\psi}$ vanishes, the network necessarily represents a non-physical quantum state, by Penrose's theorem. We give examples of this below.  
\end{example}
 
\section{Physicality of a tensor network with postselection}\label{sec:PhysTN}
First we consider questions of measurement occurrence (as in \cite{eisert2011quantum}).  Suppose a candidate state is described as a tensor network, in terms of locally physical operations.  It is natural to consider whether such a candidate state is physical.  

A library of tensors can be composed by contraction to form a class.  To prove that the physicality question is undecidable for a given class of tensor networks, we encode a known undecidable problem.  The undecidable question then maps to the constructibility of a  network with vanishing norm (evaluated by contraction with a copy of itself). 

Of course, the word problem for a finite subgroup of the group of unitary $n\times n$ matrices is decidable by simply multiplying out the words.  The matrix mortality problem is also straightforward.  
Consider the question of whether there exists a nonphysical (zero-norm) circuit built by chaining single-qutrit gates together with an initial state and a measurement.  This is shown to be undecidable in \cite{blondel2005decidable} by directly embedding Post's correspondence problem.


\begin{remark}[The word problem for linear groups is decidable]
Consider a group given as a (not necessarily finite) list of generators and relations.  {\em The word problem for groups} asks whether two words, given in the generators, represent the same group element.  Even if we restrict to groups which are finitely presented (but not finite groups), this problem is undecidable.  Thus one might hope that a faithful unitary representation would lead to an undecidable word problem for quantum circuits.  However, not all finitely presented groups (or even residually finite groups) have a faithful representation in $\GL(n,\7K)$ for some $n$ and field $\7K$.  Those that do are called {\em linear}, and have solvable word problems: given words $\sigma$, $\tau$ one can simply multiply out the matrices and check if $\sigma \tau^{-1}$ is the identity.  Alternatively, given a tensor with an input and output wire, one could bend one of the wires to form a bell state or costate.  Checking whether the corresponding tensor network state is the generalized bell state is decidable. In fact, the word problem for linear groups in characteristic zero is solvable in $\mathsf{LOGSPACE}$ \cite{lipton1977word}.
\end{remark}

This yields the following observation:
\begin{observation}
The word problem for quantum circuits built only from unitary gates (no measurements or comeasurements) is  solvable in $\mathsf{LOGSPACE}$ and the mortality problem is trivial, i.e.\ all such circuits are physical.  
\end{observation}

In contrast, in Theorem \ref{thm:QMMP} we will show that if measurement and postselection is included, the physicality problem becomes undecidable. 
First we give a simple construction that will be needed in the proof of undecidability.  Let $\|M\|_F$ be the Frobenius norm.

\begin{lemma} \label{lem:int}
Let $M$ be an $n \times n$ integer matrix.  Then $\|M\|_F^{-2}M$ can be implemented using $m=\ceil{\log_2n}$-qubit unitary gates, $m$ $\CNOT$ gates, ${m \choose 2}$ $\SWAP$ gates, and one $m$-qubit postselection operation.
\end{lemma}
The $\CNOT$ and swap gates are used to create a generalized copy tensor. 
\begin{proof}
Given an integer matrix $M$, write its SVD $M=U \Sigma V^{\top}$; we can take $U$,$V$ unitary (in fact real orthonormal) and implement $\Sigma$ by a copy tensor and postselection $\Psi$ as in Figure \ref{fig:svd}.
\begin{figure}
\[
\SelectTips{cm}{10}
\begin{xy}<5mm,0mm>:
(-4,1);
p+(0,-2) *+{M}*\frm{-} **@{-};
p+(-2,0) *+{\frac{1}{\sum_i \sigma_i^2}},
p+(2,0) *{=},
p+(0,-2) **@{-};
(0,2);
p+(0,-1.5) *+{U}*\frm{-} **@{-};
p+(0,-1.5) *{\bullet} **@{-};
{
p+(2,-2) **\crv{p+(1,0)& p+(2,-1.5)};
p+(0,-.4) *{\Psi},
p+(.6,0) **@{-};
p+(-.6,-1) **@{-};
p+(-.6,1) **@{-};
p+(.6,0) **@{-};
},
p+(0,-1.5) *+{V^\top}*\frm{-} **@{-};
p+(0,-1.5)  **@{-};
\end{xy}
\]
\caption{Quantum SVD to implement a $n \by n$ integer matrix up to global rescaling.  The wire ending in the measurement $\Psi$ has $n$ states and the measurement is $|\Psi\rangle = \frac{1}{\sum_i \sigma_i^2}\sum_i \sigma_i|i \rangle$ where $\sigma_i$ are the singular values of $M$ and the empty circle denotes a copy tensor. In the qubit case, the copy tensor is $\CNOT \circ (\bra{0} \ot \id)$.} \label{fig:svd}
\end{figure}
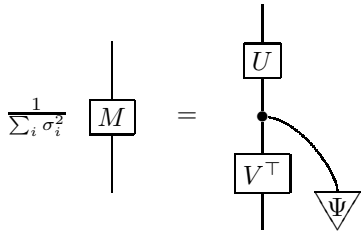
\end{proof}

For example, a two-qubit copy tensor used as a fanout can be implemented as $(\id \ot \SWAP \ot \id)\circ ((\CNOT \circ (\bra 0 | \ot \id))^{\ot 2})$.
Given a $3 \by 3$ integer matrix $M$, encode it as the two-qubit (nonunitary) operator $M'$ as follows:
\begin{small} 
\[
M= 
\begin{pmatrix}
p &0& 0\\
0 &r& 0\\
q &s& 1\\
\end{pmatrix}
\qquad M' = 
\bordermatrix{
   & \bra{00} & \bra{01} & \bra{10} & \bra{11}\cr
\ket{00} & p  & 0  & 0  & 0\cr
\ket{01} & 0  & r  & 0  & 0\cr
\ket{10} & q  & s  & 1  & 0\cr
\ket{11} & 0  & 0  & 0  & 1\cr
}.
\]
\end{small} 
It is sufficient for undecidability (\cite{paterson1970unsolvability}) to consider matrices of the form $M$, $p,r >0$, $q,s \geq 0$ together with the matrix 
\begin{small} 
\[
A= 
\begin{pmatrix}
1 &0& 1\\
-1 &0& -1\\
0 &0& 0\\
\end{pmatrix}, \qquad
A'=
\bordermatrix{
   & \bra{00} & \bra{01} & \bra{10} & \bra{11}\cr
\ket{00} & 1  & 0  & 1  & 0\cr
\ket{01} &-1  & 0  &-1  & 0\cr
\ket{10} & 0  & 0  & 0  & 0\cr
\ket{11} & 0  & 0  & 0  & 1\cr
}.
\]
\end{small} 
We have the SVD $A'=U \Sigma V^{\top}$.  So the postselection implementing $\Sigma$ will be given by $|\Psi_A\rangle = \frac{1}{5}(2|00\rangle + |01\rangle$).  
 
Now we can establish the undecidability of Question \ref{ques1} for large enough $n$ and $d$ by reducing to the matrix mortality problem.

\begin{theorem} \label{thm:QMMP}
When $n \geq 8$ and $d \geq 2$, or $n=2$, $d \geq 5$, Question \ref{ques1} is undecidable.
\end{theorem}
\begin{proof}
Question \ref{ques1} asks whether, given a set of (postselection-including) operators $M_1, M_2, \dots, M_n$, there is a finite word such that $M_{i_1}M_{i_2}\cdots M_{i_k}=0$.  The vanishing of the morphism defined by the contraction corresponding to the word means the sequence is nonphysical.

By Lemma \ref{lem:int}, we may assume our operators $M_i$ are integer matrices.  Then applying results on the integer Matrix Mortality Problem gives undecidability for our problem.  In \cite{paterson1970unsolvability}, it was established that eight $3 \times 3$ integer matrices sufficed for undecidability, and we show above how to embed this into two-qubit operators.  In \cite{cassaigne1998examples} it was shown that two $24 \times 24$ integer matrices also suffice, and these can be embedded in five-qubit operators, which gives the second part of the result.
\end{proof}

We have so far been considering decidability in tensors in which we place no restriction on the components.  It turns out that many decidability results can readily be recovered when considering the subclass of Boolean tensor networks.  These also afford an illuminating connection to solving satisfiability instances, by tensor network contraction. 

\section{Boolean tensor network states}\label{sec:Bool}
A quantum state is called Boolean iff it can be written in a local basis with amplitude coefficients taking only binary values $0$ or $1$. We relate such states with Boolean functions, allowing for a host of tools from algebra to be applied to their analysis. 

\begin{remark}[Notation]
A number in $\{0,1\}^n$ denotes an $n$-long Boolean bit string.  If $x$ is a bit string, 
then we use $\ket{x}$ as an index for a basis state.  
If $f:\{0,1\}^n\rightarrow \{0,1\}$ then $\ket{f(x)}$ also indexes a basis state.  
\end{remark}

 Let 
 \be 
 f:\{0,1\}^n \rightarrow \{0,1\}
 \ee 
 be any Boolean function.  Then 
 \be 
 \psi_{\7 B} = \sum_{\x} \ket{\x}\ket{f(\x)}
 \ee 
 is a representative in the class of Boolean states.  In this fashion, every Boolean function gives rise to a quantum state. 
A {\em Boolean relation} $R$ of arity $n$ is a subset of $\{0,1\}^n$ (for example, an \OR~clause in \SAT~can be thought of as a relation \OR~$\subset{\{0,1\}^3}$ as \OR~$=\{0,1\}^3 \setminus \{111\}$).  
Every quantum state written in a local basis with amplitude coefficients taking binary values in $\{0,1\}$ gives rise to a Boolean relation.  

\begin{theorem}[Boolean tensor network states]\label{theorem:btns}
 A tensor network representing a Boolean quantum state is determined from the classical network description of the corresponding function.  
\end{theorem}

Theorem \ref{theorem:btns} was given in \cite{biamonte2011categorical}, where the quantum tensor networks are found by letting each classical gate act on a linear space and from changing the composition of functions, to the contraction of tensors.  

Raising or lowering an index transforms kets to bras or vice versa, given a basis.  
\begin{example}[\AND-tensors]
As an example of a Boolean logic tensor, consider the \AND-tensor defined as 
$$ 
\text{\AND}^i_{jk} = \ket{00}\bra{0} + \ket{01}\bra{0} + \ket{10}\bra{0} + \ket{11}\bra{1}
$$ 
We depict the contraction of the output of the \AND-tensor with $\ket{1}$ and $\ket{0}$ as\\ 
\begin{center}
\epsfig{figure=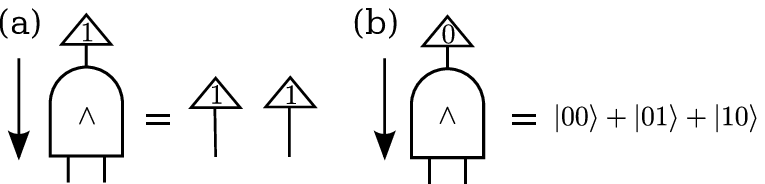, width=.45\textwidth}  
\end{center}
In (a) the contraction results in creation of the product state $\ket{11}$ and in (b) the contraction yields $\ket{00} + \ket{01} + \ket{10}$.  
\end{example}

\subsection{Satisfiability in Boolean tensor network states}
\begin{theorem}[Counting \SAT~solutions]\label{theorem:3-SAT}
Let $f$ be given to represent a \SAT~instance. Then the standard two-norm length squared can be made to give the number of satisfying assignments of the instance. 
\end{theorem}

\begin{proof}
 The quantum state takes the form 
 \be
 \psi_f = \sum_{\x} \ket{\x}\braket{f(\x)}{1} = \sum_{\x} f(\x) \ket{\x}
 \ee 
 We calculate the inner product of this state with itself viz 
 \be 
 ||\psi||^2=\sum_{\x \y} f(\x) f(\y) \langle \x, \y\rangle = \sum_\x f(\x) 
 \ee 
 which gives exactly the number of satisfying inputs.  This follows since $f(\x)f(\y)=\delta_{\x\y}$.   
 We note that for Boolean states, the square of the two norm in fact equals the one norm.  
\end{proof}

\begin{remark}[Counting \SAT~solutions]
  We note that solving the counting problem \eqref{theorem:3-SAT} is known to be {\sf \#P}-hard.
\end{remark}

\begin{corollary}[Solving \SAT~instances] 
 The condition 
 \be 
 ||\psi_f|| > 0
 \ee 
 implies that the {\sf 3-SAT} instance corresponding to $f$ has a satisfying assignment, which corresponds to an {\sf NP}-complete decision problem. 
\end{corollary}


\begin{remark}[Graphical depiction]
Below (a) gives a network realization of the function and determining if the network in (b) contracts to a value greater than zero solves a {\sf 3-SAT} instance.   
\begin{center}
 \epsfig{figure=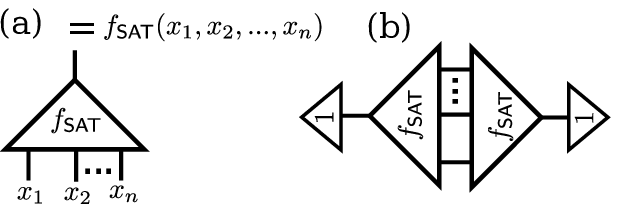,width=.4\textwidth}
\end{center}
\end{remark}

As we have considered simple restrictions on problems to transform a generally undecidable problem, into a decidable one, we can also place suitable restrictions on satisfiability problems such that they result in a class of only satisfiable instances. 

Suppose a state is defined by a tree tensor network \cite{MPSreview08,TNSreview09} such that the support of the self-contraction of tensor $T$ on each node results in the identity map (i.e.\  $\text{Supp}(TT^\dagger)=\11$).  Then the CSP defined by that state is satisfiable. 

\subsection{Decidability in Boolean tensor network states}
Now consider the {\em implementability} question: given a Boolean tensor $f$, is there some wiring of a collection $\mcF$ of Boolean tensors that produces it?  This is a word problem for planar operads, and such problems are in general undecidable.  

\begin{theorem}
Question \ref{ques2} is in general undecidable.
\end{theorem}

\begin{proof}
We rely on the proof of undecidability of Question \ref{ques2} for Boolean predicates by Cook and Bruck \cite{Cook2005implementability}.  As long as we consider only Boolean predicates in our tensor network and do not introduce swap or fanout (copy tensor), the result applies directly here. 
\end{proof}

Represent a function $x\!:\!\{1, \dots, m\} \!\ra\! R$ by a matrix with entries $x^i_j \in \{0,1\}$, $i = 1, \dots, m$, $j=1, \dots, n$.  
Given a Boolean function $f$ of arity $m$, a Boolean relation $R$ of arity $n$, and a function $x:\{1, \dots m\} \ra R$, denote by $f(x)$ the element 
$
\left ( f(x^1_1, \dots, x^m_1), f(x^1_2, \dots, x^m_2), \dots, f(x^1_n, \dots, x^m_n) \right )$
of $\{0,1\}^n.$


\begin{defn}
A Boolean relation $R$ of arity $n$ is an {\em invariant} of a Boolean function $f$ of arity $m$ if for 
any $x:[m] \ra R$, 
$f(x) \in R$. 
A Boolean function $f$ is a {\em polymorphism} of a Boolean relation $R$ if $R$ is an invariant of $f$. 
\end{defn}

If $S$ is a set of Boolean relations, let $\Pol(S)$ be the set of polymorphisms of every relation $R \in S$.  If $B$ is a set of Boolean functions, let $\Inv(B)$ be the set of invariants of every function $ f \in B$.   By \cite{bodnarchuk1969galois,geiger1968closed}, the co-clone of $S$ is $\Inv ( \Pol (S))$ and the clone of $B$ is $\Pol(\Inv(B))$.

\begin{theorem}
Let $\mcF$ be a set of Boolean tensors containing the copy tensor on three bits. Then the set of Boolean tensors constructable from $\mcF$ is $\Inv ( \Pol (\mcF))$.  
\end{theorem}
\begin{proof}
By \cite{bodnarchuk1969galois,geiger1968closed}, the co-clone of $\mcF$ is $\Inv ( \Pol (S))$. 
Thus it remains to show that the notion of co-clone of a set of relations containing the all-equal relation, and constructable tensor networks built from the corresponding Boolean tensors is the same. 
A co-clone is closed under cartesian product, identification of variables, and projection.  In the category of relations, cartesian product is the monoidal (tensor) product.  Closure under identification means connecting two wires, from the same or different tensors.  Closure under projection is equivalent to connecting a plus state to a wire of the relation; this can be implemented by identifying two legs of a three-legged all-equal relation.
\end{proof}

\begin{corollary}\label{thm:decimp}
There is a quadratic algorithm for Question  \ref{ques3}.
\end{corollary}
\begin{proof}
The set of Boolean predicates implementable by wiring predicates in $\mcF$, assuming swap and fanout, is the co-clone generated by  $\mcF$,  $\Pol(\Inv(\mcF))$.  There exists a {\em quadratic} algorithm determining if a given Boolean relation $r$ is in the co-clone generated by a finite set of relations \cite{creignou2008structure}.  
\end{proof}

The qubit case is quite special.  The equivalent implementability problem (for qudits) with $k>2$ states is co-$\mathsf{NEXPTIME}$-complete \cite{willard2010testing}.  This further illustrates how transitions among satisfiable, polynomial-time, NP- and \#P-complete, co-$\mathsf{NEXPTIME}$-complete, and undecidable problems can be caused by seemingly small changes in the allowed constructions.   

\section{Conclusion} 

The physical significance of decidability has been explored in recent work, and its relevance to quantum information science is a topic of current discussion.  We added to this discussion by focusing on decidability questions as they relate to tensor network states.  We have tried to focus here on the subtle differences between an undecidable problem, and the slight variation of this problem to transform it into a decidable one.  

JM was supported in part by the Defense Advanced Research Projects Agency under Award No. N66001-10-1-4040.  JDB completed parts of this work while visiting the Perimeter Institute for Theoretical Physics.  

\addcontentsline{toc}{section}{References}

\bibliography{bibfile}
\bibliographystyle{unsrt}

\end{document}